%% file: treeColumns.tex
\documentclass[runningheads,envcountsame]{llncs}
\input{config}

\title{Tree Drawings with Columns\thanks{J.K.\ was supported by MBIE grant UOAX1932, J.Z.\ by DFG project Wo758/11-1.}}    

\author{Jonathan~Klawitter\inst{1}\orcidID{0000-0001-8917-5269} 
\and Johannes~Zink\inst{2}\orcidID{0000-0002-7398-718X}}
\authorrunning{J. Klawitter and J. Zink}
\institute{University of Auckland, New Zealand 
\and Universität Würzburg, Germany}

\usepackage{todonotes}

\begin{document}

\maketitle

\pdfbookmark[1]{Abstract}{Abstract}
\begin{abstract}
Our goal is to visualize an additional data dimension of a tree with multifaceted data 
through superimposition on vertical strips, which we call \emph{columns}.
Specifically, we extend upward drawings of unordered rooted trees 
where vertices have assigned heights by mapping each vertex to a column.
Under an orthogonal drawing style and with every subtree within a column drawn planar,
we consider different natural variants concerning the arrangement of subtrees within a column.
We show that minimizing the number of crossings in such a drawing
can be achieved in fixed-parameter tractable (FPT) time 
in the maximum vertex degree~$\Delta$ for the most restrictive variant,
while becoming NP-hard (even to approximate)
already for a slightly relaxed variant.
However, we provide an FPT algorithm 
in the number of crossings plus~$\Delta$, and an FPT-approximation algorithm in~$\Delta$
via a reduction to 
feedback arc set. 

\end{abstract}

\section{Introduction} 
\label{sec:intro}
Visualizations of trees have been used for centuries, 
as they provide valuable insights into the structural properties 
and visual representation of hierarchical relationships~\cite{Lim14}. 
Over time, numerous approaches have been developed to create tree layouts 
that are both aesthetically pleasing and rich in information.
These developments have extended beyond displaying the tree structure alone
to also encompass different \emph{facets} (dimensions) of the underlying data. 
As a result, researchers have explored various layout styles beyond layered node-link diagrams
and even higher-dimensional representations~\cite{treevis,Rus13}.
Hadlak, Schumann, and Schulz~\cite{HSS15} introduced the term of a \emph{multifaceted graph}
for a graph with associated data that combines various facets (also called aspects or dimensions) 
such as spatial, temporal, and other data.  
An example of a multifaceted tree is a \emph{phylogenetic tree},
that is, a rooted tree with labeled leaves 
where edge lengths represent genetic differences or time estimates
by assigning heights to the vertices. 
In this paper, we introduce an extension of classical node-link drawings for rooted phylogenetic trees, 
where vertices are mapped to distinct vertical strips, referred to as \emph{columns}. 
This accommodates not only an additional facet of the data 
but also introduces the possibility of crossings between edges, 
thus presenting us with new algorithmic challenges.

\paragraph{Motivation.}
Our new drawing style is motivated by the visualization of transmission trees 
(i.e., the tree of who-infected-who in an infectious disease outbreak),
which are phylogenetic trees that often have rich multifaceted data associated.
In these visualizations, geographic regions associated with each case 
are commonly represented by coloring the vertex and its incoming edge~\cite{Nextstrain}.
However, there are scenarios where colors are not available or suitable.
In such a case and in the context of a transmission tree, 
our alternative approach maps each, say, region or age group to a separate column (see~\cref{fig:example}).
The interactive visualization platform Nextstrain~\cite{Nextstrain} for tracking of pathogen evolution
partially implements this with a feature that allows the mapping of the leaves to columns based on one facet.
Yet, since they then leave out inner vertices and edges, the topology of the tree is lost.
We hope that our approach enables users to quickly grasp group sizes and
the number of edges (transmissions) between different groups.

\begin{figure}[t]
  \centering
  \includegraphics{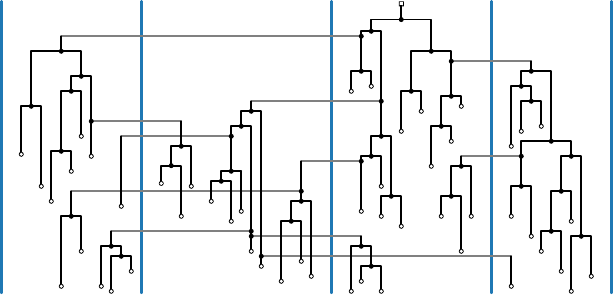}
  \caption{A binary column tree with four columns.}
  \label{fig:example}
\end{figure}

\paragraph{Related Work.}
In visualization methods for so-called reconciliation trees~\cite{CDSJJB16,CDMP20} 
and multi-species coalescent trees~\cite{Dou20,KKNW23},
a guest tree $P$ is drawn inside a space-filling drawing of a host tree $H$.
The edges of the host tree can thus closely resemble columns,
though the nature of and relationship between $P$ and $H$ 
prohibit a direct mapping to column trees.
Spatial data associated with a tree has also been visualized by juxtaposition
for so-called phylogeographic trees~\cite{PMZ+13,KKSvDZ23}.
Betz \etal~\cite{BGMRW14} investigated orthogonal drawings with column assignments,
where each column is a single vertical line (i.e. analogous to a layer assignment). 

\paragraph{Setting.}
The input for our drawings consists of an unordered rooted tree~$T$, 
where each vertex $v \in V(T)$ has an assigned \emph{height}~$h(v)$,
and a surjective \emph{column mapping} $c \colon V(T) \to [1, \ell]$ for some $\ell \geq 2$.
Together, we call $\croc{T, h, c}$ a \emph{column tree}.
We mostly assume that the order of the columns is \emph{fixed} (from~1 to~$\ell$ left-to-right),
but in a few places we also consider the case that it is \emph{variable}.
For an edge~$uv$ in~$T$, $u$ is the \emph{parent} of~$v$, 
and~$v$ is the \emph{child} of~$u$.
We call~$u$ and~$v$ the \emph{source} and \emph{target} (vertex) of~$uv$, respectively.
We call $uv$ an \emph{intra-edge} if $c(u) = c(v)$ and an \emph{inter-edge} otherwise.
The degree of~$u$ 
is the number of children of~$u$.

\paragraph{Visualization of Column Trees.}
We draw a column tree $\croc{T, h, c}$ with a {\em rectangular cladogram} style,  
that is, each edge is drawn orthogonally and (here) downward with respect to the root;
hence we assume that~$h$ corresponds to y-coordinates 
where the root has the maximum value and every parent vertex has a strictly greater value than its children.
Each edge $uv$ has at most one bend such that the horizontal segment (if existent)
has the y-coordinate of the parent vertex~$u$.
Each column $\gamma \in [1, \ell]$ is represented by a vertical strip~$C_\gamma$ of variable width
and a vertex~$v$ with $c(v) = \gamma$ must be placed within~$C_\gamma$.
We need a few definitions to state further drawing conventions. 

A \emph{column subtree} is a maximal subtree within a column.
Note that each column subtree (except the one containing the root of $T$) 
has an incoming inter-edge to its root
and may have various outgoing inter-edges to column subtrees in other columns.
The \emph{width} of a column subtree at height $\eta$ is the number of edges
of the column subtree intersected by the horizontal line at $\mathrm{y} = \eta$.
We say that two column subtrees $A$ and $B$ \emph{interleave} in a drawing of a column tree
if there is a horizontal line that intersects first an edge of $A$, 
then one of $B$, and then again one of $A$ (or with $A$ and $B$ in reversed roles).

\begin{figure}[b]
	\centering
	\begin{subfigure}[t]{0.35 \linewidth}
		\centering
		\includegraphics[page=1]{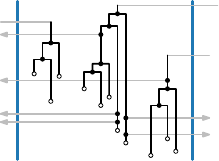}
		\caption{V1 -- subtree ``stick'' to column borders (13 crossings)}
		\label{fig:variants:one}
	\end{subfigure}
	\hfill
	\begin{subfigure}[t]{0.3 \linewidth}
		\centering
		\includegraphics[page=2]{variants}
		\caption{V2 -- non-interleaving subtrees (11 cr.)}
		\label{fig:variants:two}
	\end{subfigure}
	\hfill
	\begin{subfigure}[t]{0.3 \linewidth}
		\centering
		\includegraphics[page=3]{variants}
		\caption{V3 -- interleaving subtrees allowed (6 cr.)}
		\label{fig:variants:three}
	\end{subfigure}
	\caption{The three variants for subtree arrangements.} 
	\label{fig:variants}
\end{figure}

In graph drawing, we usually 
forbid overlaps.
Here, however, if a vertex has more than one child to, say, its right, 
then the horizontal segments of the edges to these children overlap.
As we permit more than three children,
we allow these overlaps and we do not count them as crossings.
On the other hand, we do not want overlaps to occur
between non-neighboring elements.
Thus, we assume 
that each vertex~$v$ being the source
of an inter-edge has a unique height $h(v)$
as otherwise we cannot always avoid overlaps
between an edge and a vertex and between two edges with distinct endpoints.
Furthermore, we 
require that {\em no two intra-edges cross}. 
Hence, each column subtree is drawn planar and no two column subtrees intersect.
However, since for inter-edges it is not always possible to avoid all edge crossings,
we distinguish three drawing conventions concerning inter-edges and column subtree relations.
These increasingly trade clarity of placement 
with the possibility to minimize the number of crossings~(see~\cref{fig:variants}):
\begin{enumerate}[label=(V\arabic*),leftmargin=*]
  \item No inter-edge $uv$ intersects an intra-edge in column $c(v)$.
  \item No two column subtrees interleave.
  \item Column subtrees may interleave.
\end{enumerate}
We remark that V1 is based on the idea that going from top to bottom, 
the column subtree rooted at $v$ is greedily placed as soon as possible in $c(v)$. 
Thus, the subtrees ``stick'' to the column border closer to their root's parent.
Note that V1 implies V2 as interleaving a column subtree $B$ inside a column subtree $A$
requires that the incoming inter-edge of $B$ crosses an intra-edge of~$A$. 

\paragraph{Combinatorial Description.}
We are less interested in the computation of x-coordinates
but in the underlying algorithmic problems regarding the number of crossings.
Hence, it suffices to describe a column tree drawing combinatorially in terms of what we call 
the \emph{subtree embeddings} and the \emph{subtree arrangement} of the drawing.
The subtree embedding is determined by the order of the children of each vertex in the subtree. 
The subtree arrangement of a column represents the relative order between column subtrees that overlap vertically.
From a computational perspective, one possible approach to specify the subtree arrangement is by storing, 
for each root vertex $v$ of a column subtree, which edges of which column subtrees lie horizontally to the left and to the right of $v$.
This requires only linear space and all necessary information, 
e.g., for an algorithm that computes actual x-coordinates, can be recovered in linear time.
A fixed column order, a subtree embedding of each column subtree, 
and a subtree arrangements of each column 
collectively form what we refer to as an \emph{embedding} of a column tree~$\croc{T, h, c}$.

\paragraph{Problem Definition.}
An instance of the problem \tcv{$i$}, $i \in \set{1, 2,3}$, is a column tree $\croc{T, h, c}$.
The task is to find an embedding of $\croc{T, h, c}$
with the minimum total number of crossings among all
possible embeddings of~$T$ under the drawing convention V$i$.
In the decision variant, 
an additional integer $k$ is given, 
and the task is to find a column embedding with at most $k$ crossings.
If the column order is not fixed, we get another three problem versions.

We define three types of crossings.
Consider an inter-edge $uv$ with $c(u) <~c(v)$.
In any drawing, $uv$ spans over the columns $c(u) + 1, \dots, c(v) - 1$
and crosses always the same set of edges in these columns.
We call the resulting crossings \emph{inter-column} crossings,
denoted by~$k_\mathsf{inter}$.
Within the column $c(u)$ [$c(v)$],
$uv$ may cross edges of the column subtree of $u$ [$v$]
and of other column subtrees. 
We call these crossings \emph{intra-subtree} and \emph{intra-column} crossings,
denoted by $k_\mathsf{subtree}$ and~$k_\mathsf{column}$, respectively.
The total number of crossings is then
$k = k_\mathsf{subtree} + k_\mathsf{column} + k_\mathsf{inter}$.

It would also be natural to seek an embedding or drawing with minimum width,
yet this is known to be NP-hard even for a single phylogenetic tree~\cite{BVGJO19}.

Observe that for \tcv{1} and V2, 
modifying the embedding of a single column subtree
can only change the intra-subtree crossings.
On the other hand, permuting the order of the column subtrees
within a column can only change the intra-column crosssings.
Finding a minimum-crossing embedding can thus be split into two tasks,
namely, \emph{embedding the subtrees} and \emph{finding a subtree arrangement}.
Furthermore, note that V1 mostly enforces a subtree arrangement
since in general each column subtree is the left-/rightmost in its column
at the height of its root. 
Only for column subtrees that are in the same column and 
whose roots share a parent, is the relative order not known.
The main task for V1 
is thus to find subtree embeddings.
For V3, in contrast, the two tasks are intertwined
since a minimum-crossing embedding may use locally suboptimal subtree embeddings
(see \cref{fig:variants:three}).
We briefly discuss heuristics 
in \cref{sec:hardest}.

\paragraph{Contribution.}
We first give a fixed-parameter tractable (FPT) algorithm in the maximum vertex degree~$\Delta$ 
for the subtree embedding task, the main task for V1 and a binary column tree (\cref{sec:easy}).
We remark that the previously described applications usually use binary phylogenetic trees.
This makes our algorithm, which uses a sweep line
to determine the crossing-minimum child order at each vertex, a practically relevant polynomial-time algorithm.
On the other hand, if we have a large vertex degree,
minimizing the number of crossings is NP-hard,
even to approximate, for all variants (\cref{sec:hard}).
This holds true for the less restrictive variants V2 and V3 even for binary trees.
Leveraging a close relation between the task of finding a column subtree arrangement and the feedback arc set (\fas) problem,
we devise for V2 an FPT algorithm in the number of crossings $k$ plus $\Delta$
and an FPT-approximation algorithm in $\Delta$ (\cref{sec:algos}).
However, for the most general variant V3, the tasks of embedding the column subtrees
and arranging them cannot be considered separately,
and hence we only suggest heuristics to address this challenge (\cref{sec:hardest}).

\section{Algorithm for Subtree Embedding} 
\label{sec:easy}
In this section, we describe how to find an optimal subtree embedding 
for a single column subtree. 
The algorithm has an FPT running time in the maximum vertex degree~$\Delta$
making it a polynomial-time algorithm for bounded-degree trees
like binary trees, ternary trees, etc.

Our algorithm is based on the following observation,
described here for a binary column subtree~$S$.
Fix an embedding for~$S$, i.e., the child order for each vertex.
Consider an inter-edge~$e = vw$ with~$c(w) < c(v)$. 
Let~$u$ be a vertex of the path from the parent of~$v$ to the root of~$B$. 
Observe that if~$v$ is in the left subtree of~$u$, 
then~$e$ does not intersect the right subtree of~$u$.
Yet, if~$v$ is in the right subtree of~$u$,
then~$e$ intersects the left subtree of~$u$ (if it extends vertically beyond~$h(v)$).
On the other hand, if~$u$ is not on the path to the root,
then the child order of~$u$ has no effect on whether~$e$ intersects its subtrees.
Thus, to find an optimal child order of a vertex,
it suffices to consider the directions (left/right)
towards which the inter-edges of its descendants extend.
Our algorithm is illustrated in \cref{fig:algo}.

\begin{figure}[t]
	\centering
	\begin{subfigure}[t]{0.3 \linewidth}
		\centering
		\includegraphics[page=2]{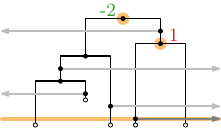}
		\caption{Initial embedding and first inter-edge.}
		\label{fig:algo:one}
	\end{subfigure}
	\hfill
	\begin{subfigure}[t]{0.3 \linewidth}
		\centering
		\includegraphics[page=3]{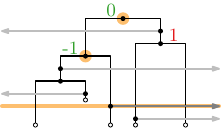}
		\caption{The second inter-edge.}
		\label{fig:algo:two}
	\end{subfigure}
	\hfill
	\begin{subfigure}[t]{0.35 \linewidth}
		\centering
		\includegraphics[page=4]{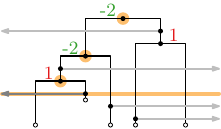}
		\caption{The third inter-edge.}
		\label{fig:algo:three}
	\end{subfigure}
	
	\bigskip
	
	\begin{subfigure}[t]{0.3 \linewidth}
		\centering
		\includegraphics[page=5]{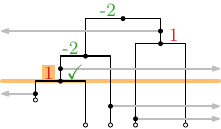}
		\caption{Keeping a child order to avoid one crossing.}
		\label{fig:algo:four}
	\end{subfigure}
	\hfill
	\begin{subfigure}[t]{0.3 \linewidth}
		\centering
		\includegraphics[page=6]{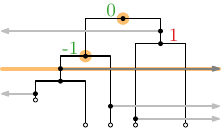}
		\caption{The fourth inter-edge.}
		\label{fig:algo:five}
	\end{subfigure}
	\hfill
	\begin{subfigure}[t]{0.35 \linewidth}
		\centering
		\includegraphics[page=7]{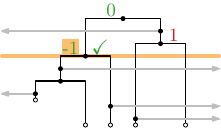}
		\caption{Keeping a child order to avoid one crossing.}
		\label{fig:algo:six}
	\end{subfigure}

	\bigskip
	
	\begin{subfigure}[t]{0.3 \linewidth}
		\centering
		\includegraphics[page=8]{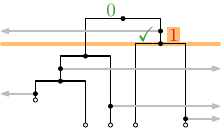}
		\caption{Changing a child order to avoid one crossing.}
		\label{fig:algo:seven}
	\end{subfigure}
	\hfill
	\begin{subfigure}[t]{0.3 \linewidth}
		\centering
		\includegraphics[page=9]{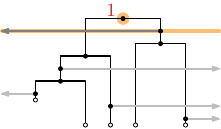}
		\caption{The fifth inter-edge.}
		\label{fig:algo:eight}
	\end{subfigure}
	\hfill
	\begin{subfigure}[t]{0.35 \linewidth}
		\centering
		\includegraphics[page=10]{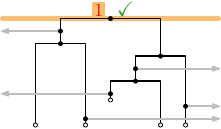}
		\caption{Changing the child order of the root to avoid one crossing.}
		\label{fig:algo:nine}
	\end{subfigure}
	\caption{Steps of the sweep-line algorithm to find, for a given (here, binary) column subtree,
		a subtree embedding with the minimum number of intra-subtree crossings.
		A number at a vertex $v$ reflects the count of crossings
		for the two possible permutations of $v$
		(positive number $\Leftrightarrow$ more crossings if $v$'s child order~is~kept).
		In steps considering inter-edges, affected ancestor vertices are marked orange.} 
	\label{fig:algo}
\end{figure}

\begin{lemma} \label{clm:algo}
For an $n$-vertex column subtree $S$ with maximum degree $\Delta$ and~$t$ inter-edges,
a minimum-crossing subtree embedding of $S$ can be computed in $\Oh(\Delta! \Delta t n + n \log n)$~time.
\end{lemma}
\begin{proof}
In a bottom-to-top sweep-line approach over $S$, 
we apply the following greedy-permuting strategy 
that finds, for each vertex, the child order that causes the fewest crossings.
To this end, first sort the
vertices of $S$ by ascending height in $\Oh(n \log n)$ time.
Also, each inner vertex $v$ with $d_v$ ($d_v \le \Delta$) children in~$S$
has $d_v!$ counters to store for each of the $d_v!$ possible child orders 
(of its children in~$S$) the computed number of crossings induced by that order (as described below).
Initially, all counter are set to zero.

Having started at the bottom, let $v$ be the next encountered vertex.
First, process each inter-edge~$e$ whose source is $v$.
Let $u$ be a vertex on the path~$P$ from the parent of $v$ to the root. 
Let $T_v$ be the subtree rooted at a child of~$u$ that contains~$v$.
Compute for each subtree $T$ rooted at a child of $u$ (except for~$T_v$) the width at $h(v)$ as follows.

Initialize a counter for the width of $T$ at $h(v)$ with zero.
Traverse $T$ and whenever we encounter an edge of~$T$
that has one endpoint above and the other endpoint on or below $h(v)$,
we increment that counter.
Hence, we can determine the width of $T$ at $h(v)$
in $\Oh(n_T)$ time, where $n_T$ is the number of vertices in~$T$,
and we can determine the width of all subtrees rooted at children
of vertices from $P$ in $\Oh(n)$ time.

Recall that $u$ has $d_u!$ (with $d_u \le \Delta$)
counters for its $d_u!$ possible child orders.
For each such child order, add,
if $e$ goes to the left [right],
the width at $h(v)$ of all subtrees rooted at children of $u$
that are left [right] of $T_v$ to the corresponding counter.
Note that the resulting number is exactly
the number of intra-subtree crossings induced by $e$,
as our observation still holds for the general case
because the number of intra-subtree crossings induced by~$e$
only depends on the child orders of the vertices of~$P$.
Of course, as the numbers are independent of other inter-edges,
we can add up these numbers in our child order counters.
Updating these counters for a vertex~$u$ on~$P$
takes $\Oh(d_u! d_u) \subseteq \Oh(\Delta!\Delta)$ time.
Since the length of $P$ can be linear,
we can handle $e$ in $\Oh(\Delta!\Delta n)$ time.

Second, since the counters of $v$ only depend on
the already processed inter-edges below $v$, 
pick the child order of $v$ with the lowest counter for~$v$.

%

In total, this sweeping phase of the algorithm runs in $\Oh(\Delta! \Delta t n)$ time,
where $t$ is the number of inter-edges.
This is because, although the sweep-line algorithm has $\Theta(n)$
event points (the vertices of~$S$), we apply the
previously described $\Oh(\Delta!\Delta n)$-time subroutine
only if we encounter an inter-edge.
\end{proof}

Note that for \cref{clm:algo}, if $t = 0$, any embedding is crossing free,
and if $\Delta$ is constant, the running time is at most quadratic in $n$.

To solve an instance of \tcv{1},		
we first apply \cref{clm:algo} to each column subtree separately.
Then for each vertex with multiple outgoing inter-edges to the same column,
we try all possible orders for the respective column subtrees
and keep the best. Hence, we get the following result.

\begin{theorem} \label{clm:easy}
	\tcv{1} is fixed-parameter tractable in~$\Delta$.
	More precisely, given an instance $\croc{T, h, c}$
	with $n$ vertices and maximum vertex degree~$\Delta$,
	there is an algorithm computing an embedding of $\croc{T, h, c}$
	with the minimum number of crossings in $\Oh(\Delta!\Delta n^2)$ time.
\end{theorem}

\section{NP-Hardness and APX-Hardness} 
\label{sec:hard}
In this section, we show that \tc becomes NP-hard
(even to approximate) when finding a subtree arrangement
is non-trivial or if we have large vertex degrees.
We use a reduction from the (unweighted) \textsc{Feedback Arc Set} (\fas) problem,
where we are given a digraph $G$ and the task is to find
a minimum-size set of edges that if removed make $G$ acyclic.
We assume, without loss of generality, that~$G$ is biconnected.
\fas is one of Karp's original 21 NP-complete problems~\cite{Karp72}.
It is also NP-hard to approximate within a factor
less than~$1.3606$~\cite{DS05} and,
presuming the unique games conjecture, even within any constant factor~\cite{GHMRC11}.

\begin{theorem} \label{clm:np}
The \tc problem is NP-complete for both fixed and variable column orders
already for two columns
under V1 if the degree is unbounded, and under V2 and V3 even for a binary column tree.
Moreover, it is NP-hard to approximate within a factor
of $1.3606 - \varepsilon$ for any $\varepsilon > 0$,
and NP-hard to approximate within any constant factor presuming
the unique \mbox{games conjecture}.
\end{theorem}
\begin{proof}
The problem is in NP, since given an embedding for a column tree, 
it is straightforward to check whether it has at most~$k$ crossings in polynomial~time.

To prove NP-hardness (of approximation),
we use a reduction from \fas.
Let~$G$ be an instance of \fas;
let~$n = \abs{ V(G) }$ and~$m = \abs{ E(G) }$, and
fix an arbitrary vertex order via the indices, so~$V(G) = \set{v_1, \ldots, v_n}$,
and an arbitrary edge order.
We construct a column tree~$\croc{T, h, c}$
such that there is a bijection between
vertex orders of~$G$ (where each vertex order implies a solution set
for the \fas problem)
and solutions of~$\croc{T, h, c}$,
whose number of crossings depends on
the size of the \fas solution and vice versa.
We show this first for V1 with unbounded degree;
see~\cref{fig:np}.

Give~$T$ two columns and, assuming for now that their order is fixed, 
call them left and right column.
The total height of~$T$ is~$a = 3 (m + 1)$.
The root column subtree of~$T$ is in the left column,
and has a single inner vertex $v$ at $a - 1$. 
For each~$v_i \in V(G)$, there is an inter-edge from $v$
to a column subtree~$T_i$, the \emph{vertex gadget} of~$v_i$, in the right column.
The idea is that a subtree arrangement
of $\{T_1, \dots, T_n\}$
corresponds to a topological order of~$G \setminus S$
where each edge whose \emph{edge gadget} induces a large number of
crossings is in the \fas solution set~$S$.

\paragraph{Vertex Gadget.}
The subtree~$T_i$ has a \emph{backbone} path on~$\deg(v_i) + 1$ vertices
with its leaf at height 0.
Attached to the backbone of~$T_i$,
there are inter-edges and subtrees depending
on the edges incident to~$v_i$.
We describe them next.

\paragraph{Edge Gadget.}
The edge gadget for an edge~$v_iv_j \in E(G)$ consists of
an inter-edge~$e_{i,j}$ to the left column
attached to the backbone of~$T_i$
and of a star~$S_{i,j}$ with~$n^3$ leaves attached to the backbone of~$T_j$.
Both are contained in a horizontal strip of height three.
So if~$v_iv_j$ is, say, the~$t$-th edge,
then we use the heights between~$b = a - 3t$ to~$b - 2$.
The root of~$S_{i,j}$ is at height~$b$
and the leaves of~$S_{i,j}$ are at height~$b-2$.
The source and the target of $e_{i,j}$
are at heights $b-1$ and $b-2$, respectively.
Note that every~$T_i$ (on its own) can always be drawn planar.

\begin{figure}[t]
  \centering
  \begin{subfigure}[t]{0.2 \linewidth}
	\centering
    \includegraphics[page=1]{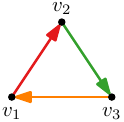}
	\caption{A FAS instance on three edges.}
	\label{fig:np:fas}
  \end{subfigure}
  \hfill
  \begin{subfigure}[t]{0.75 \linewidth}
	\centering
    \includegraphics[page=2]{np}
	\caption{In the corresponding column subtree for V1, 
	each edge gadget occupies a horizontal strip and 
	we need to find a subtree arrangement of the vertex gadgets (grey background).} 
	\label{fig:np:reduction}
  \end{subfigure}
  \caption{Reduction of a \fas instance to a \tc instance.}
  \label{fig:np}
\end{figure}

\paragraph{Analysis of Crossings.}
Note that, for each~$v_iv_j \in E(G)$, 
the inter-edge~$e_{i,j}$ induces at most~$n - 1$ crossings
with the backbones of~$\{T_1, \dots, T_n\}$
independent of the subtree arrangement.
With~$m \leq n(n-1)/2$, there can thus be at most~$n(n-1)/2 \cdot (n-1) = n^3/2 - n^2 + n/2$ 
crossings that do not involve a star subtree.
Now, if~$T_i$ is to the left of~$T_j$, 
then~$e_{i, j}$ does not intersect~$S_{i,j}$.
On the other hand, if~$T_i$ is to the right of~$T_j$,
then~$e_{i,j}$ intersects~$S_{i, j}$,
which causes~$n^3$ crossings.
So, 
for a subtree arrangement with~$s$ edge gadgets causing crossings,
the embedding of~$T$ contains~$sn^3$ plus at most~$rn^3$
crossings, where $r = 1/2 - 1/n + 1/(2n^2) \le 1$.

\paragraph{Bijection of Vertex Orders and Subtree Arrangements.}
As each columns subtree~$T_i$ ($i \in \{1, \dots, n\}$)
corresponds to vertex~$v_i$ of $G$,
each subtree arrangement of $\croc{T, h, c}$
implies exactly one vertex order of $G$ and vice versa.
A vertex order~$\Pi$ of $G$, in turn, implies
the \fas solution set containing all edges
whose target precedes its source in~$\Pi$.
In the other direction, for a \fas solution set~$S$,
we find a corresponding vertex order 
by computing a topological order of $G \setminus S$.

Consequently, a subtree arrangement of~$\croc{T, h, c}$
whose number~$k$ of crossings lies in $[sn^3, (s+r)n^3]$
implies a \fas solution of size~$s$ and vice versa.
In particular, a minimum-size \fas solution of size~$s^*$
corresponds to an optimal subtree arrangement
with $k_{\min} = (s^*+r^*) n^3$ crossings,
where $0 \le r^* \le r \le 1$.

\paragraph{Hardness of Approximation.}
If it is NP-hard to approximate \fas
by a factor less than~$\alpha$,
then it is NP-hard to approximate \tcv{1} by
a factor less than~$\beta$ due to the previous bijection,
where $\beta$ can be bounded as follows.
\begin{equation*}
	\beta \cdot k_{\min} = \beta \cdot (s^*+r^*) n^3 \ge \alpha s^* n^3
	\quad \Leftrightarrow \quad
	\beta \ge \frac{\alpha s^*}{s^*+r^*}
	= \alpha - \frac{\alpha r^*}{s^*+r^*}
\end{equation*}
Since $s^*$ is unbounded, $\alpha r^*/(s^*+r^*)$
can be arbitrarily close to zero.
Hence, it is NP-hard to approximate \tcv{1} with unbounded maximum degree
within any factor $\alpha - \varepsilon$ for any $\varepsilon > 0$.

\begin{figure}[b]
	\centering
	\includegraphics[page=3]{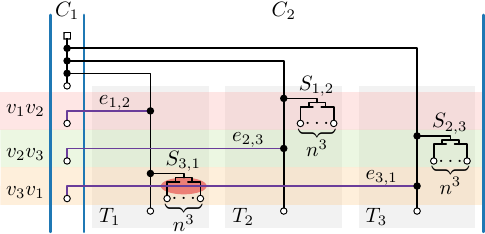}
	\caption{Reduction for V2 \& V3 with binary trees.} 
	\label{fig:np:binary}
\end{figure}

\paragraph{Other Variants.}
For a variable column order, note that, for our bijection,
the vertex order is the same but mirrored
if we swap the left and right column.

For a binary column tree under V2 and V3, 			
we let the sources of the inter-edges in the root column subtree,
which have as targets the roots of the $T_i$, form a path in the root column subtree; see \cref{fig:np:binary}.
These edges can form at most~$n^2$ pairwise crossings,
which changes $r$ and $r^*$ only slightly and
has thus no effect on our bijection.
Similarly, the star subtrees can simply be substituted with binary subtrees 
where all inner vertices have a height between $b$ and $b - 1$.
Also note that the $T_i$ cannot interleave
since they span from above the first edge gadgets all the way to the bottom.
\end{proof}

Note that in the proof of \cref{clm:np}, 
the variant with binary columns trees does not work under V1 (see~\cref{fig:np:binary}),
since then it would no longer be possible to permute the $T_i$;
there would only be one possible subtree arrangement.


\section{Algorithms for Subtree Arrangement} 
\label{sec:algos}
In the previous section, we have seen that \tcv{2} is NP-hard,
even to approximate, by reduction from \fas.
Next, we show that we can also go the other way around
and reduce \tcv{2} to \fas to obtain
an FPT algorithm in the number of crossings and the maximum vertex degree~$\Delta$,
and to obtain an FPT-approximation algorithm in $\Delta$.

\paragraph{Integer-Weighted Feedback Arc Set.}
As an intermediate step in our reduction,
we use the \textsc{Integer-Weighted Feedback Arc Set} (\ifas) problem
which is defined for a digraph $G$ as \fas but with the additional property
that each edge $e$ has a positive integer weight $w(e)$,
$w \colon E(G) \rightarrow \mathbb{N}$,
and the objective is to find a minimum-weight set of edges
whose removal results in an acyclic graph.
Clearly, \ifas is a generalization of \fas
because \fas is \ifas with unit weights.
However, an instance $\croc{G, w}$ of \ifas can also straightforwardly be expressed as
an instance $G'$ of \fas:
Obtain~$G'$ from~$G$ by substituting each edge $uv$ 
with $w(uv)$ length-two directed paths from~$u$ to~$v$.
For an illustration, see \cref{fig:ifas}.

\begin{figure}[tbh]
	\centering
	\begin{subfigure}[t]{.45 \linewidth}
		\centering
		\includegraphics[page=1]{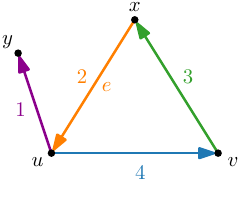}
		\caption{Input \ifas instance $\croc{G,w}$ with minimum-weight solution $S = \set{e}$.}
		\label{fig:ifas:input}
	\end{subfigure}
	\hfill
	\begin{subfigure}[t]{.46 \linewidth}
		\centering
		\includegraphics[page=2]{reduction}
		\caption{Output \fas instance $G'$ with mini\-mum-size solution $S' = \set{e_1, e_2}$.}
		\label{fig:ifas:output}
	\end{subfigure}
	\caption{Reduction of an \ifas instance to a \fas instance.}
	\label{fig:ifas}
\end{figure}

\begin{lemma} \label{clm:ifas2fas}
	We can reduce an instance $\croc{G,w}$ of \ifas with $m$ edges
	and maximum weight $w_{\max}$
	to an instance $G'$ of \fas in time $\Oh(mw_{\max})$ such that 
	\begin{itemize}
		\item $G'$ has size in $\Oh(mw_{\max})$,
		\item the size of a minimum-size solution of $G'$ equals
		the weight of a minimum-weight solution of $\croc{G,w}$, and
		\item we can transform a solution of $G'$ with size $s$ in $\Oh(s)$ time
		to a solution of~$G$ with weight at most $s$.
	\end{itemize}
\end{lemma}
\begin{proof}
	Note that, by construction,
	$G'$ has size in $\Oh(mw_{\max})$ and can be construct in $\Oh(mw_{\max})$~time.
	
	Consider any minimum-weight solution~$S$ of~$\croc{G,w}$ with total weight~$s$.
	We can transform it to a solution~$S'$ of~$G'$ of size $s$
	by removing, for each edge $uv$ of~$S'$,
	all first edges of the $w(uv)$ length-two paths representing $uv$ in~$G'$.
	If $G' \setminus S'$ had a cycle~$C$,
	then a cycle half the length of $C$ would also occur in $G \setminus S$.
	
	For any solution~$S'$ of~$G'$ with size~$s$,
	there is also a solution~$S$ of $\croc{G,w}$ with total weight at most~$s$:
	For an edge $uv$ in $G$, consider the $w(uv)$ length-two paths representing $uv$ in $G'$.
	If $S'$ contains at least one edge of every such path,
	we add the edge $uv$ to the solution $S$.
	Clearly, the weight of~$S$ is
	at most the number of edges in~$S'$.
	Suppose for a contradiction that there is a cycle~$C$ in $G \setminus S$.
	Then, there would also be a cycle in $G' \setminus S$
	because for every edge of~$C$,
	there is at least one length-two path left in $G'$ after removing $S'$.
\end{proof}

\paragraph{Reduction to Feedback Arc Set.}
Next, we show that we can express every instance of \tcv{2}
as an instance of \ifas.
Note, however, that we have split the problem into the tasks
of embedding the subtrees and finding a subtree arrangement.
Here, we are only concerned with finding a subtree arrangement
for every column, and we assume that we separately solve
the problem of embedding every column subtree,
e.g., by employing the algorithm from \cref{clm:algo}.

\begin{lemma} \label{clm:tc2ifas}
	We can reduce an instance $\croc{T, h, c}$ of \tcv{2} on $n$ vertices
	to an instance $\croc{G,w}$ of \ifas in time $\Oh(n^2)$ such that 
	\begin{itemize}
		\item $G$ has size in $\Oh(n^2)$ and maximum weight $w_{\max} \in \Oh(n^2)$,
		\item the number of intra-column crossings in a minimum-crossing solution of $\croc{T, h, c}$ equals
		the weight of a minimum-weight solution of $\croc{G,w}$
		plus $t$ where $t$ is some integer in $\Oh(n^2)$
		depending only on $\croc{T, h, c}$, and
		\item we can transform a solution of $\croc{G,w}$ with weight $s$ in $\Oh(n^2)$ time
		to a solution of~$\croc{T, h, c}$ with at most $s + t$ intra-column crossings.
	\end{itemize}
\end{lemma}

\begin{figure}[t]
	\centering
	\begin{subfigure}[t]{.5 \linewidth}
		\centering
		\includegraphics[page=3]{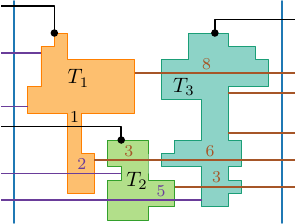}
		\caption{Possible permutation $\Pi$ of the subtrees.}
		\label{fig:tc2ifas:tcnormal}
	\end{subfigure}
	\hfill
	\begin{subfigure}[t]{.45 \linewidth}
		\centering
		\includegraphics[page=4]{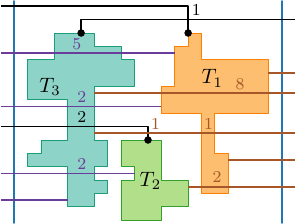}
		\caption{Reverse permutation of $\Pi$.}
		\label{fig:tc2ifas:tcinverse}
	\end{subfigure}
	
	\bigskip
	
	\begin{subfigure}[t]{.44 \linewidth}
		\centering
		\includegraphics[page=5]{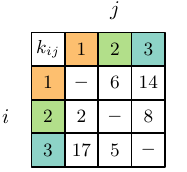}
		\caption{Table of pairwise intra-column crossings $k_{ij}$ occurring if subtree $T_i$ is placed to the left of subtree~$T_j$.}
		\label{fig:tc2ifas:table}
	\end{subfigure}
	\hfill
	\begin{subfigure}[t]{.45 \linewidth}
		\centering
		\includegraphics[page=6]{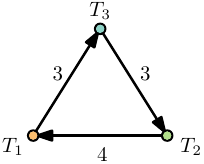}
		\caption{Resulting \ifas instance where the difference $k_{ij} - k_{ji}$
			determines the direction and the weight of an edge.}
		\label{fig:tc2ifas:output}
	\end{subfigure}
	
	\caption{Reduction of an \tcv{2} instance to an \ifas instance.}
	\label{fig:tc2ifas}
\end{figure}

\begin{proof}
	To each column in $c$, we apply the following reduction;
	see \cref{fig:tc2ifas}.
	Let $T_1, \dots, T_r$ be the set of column subtrees of $c$.
	For each pair $T_i, T_j$ of subtrees
	($i, j \in \set{1, \dots, r}, i \ne j$),
	we consider the number of intra-column crossings
	where only edges with an endpoint in $T_i$ and $T_j$
	are involved~--
	once if $T_i$ is placed to the left of $T_j$ and
	once if $T_j$ is placed to the left of $T_i$;
	see~\cref{fig:tc2ifas:tcnormal,fig:tc2ifas:tcinverse}.
	We let these numbers be $k_{ij}$ and~$k_{ji}$, respectively;
	see~\cref{fig:tc2ifas:table}.
	We then construct an \ifas instance $\croc{G, w}$
	where~$G$ has a vertex for every column subtree.
	For each pair $T_i, T_j$,
	$G$ has the edge $T_iT_j$ if $k_{ij}-k_{ji} < 0$, or
	$G$ has the edge $T_jT_i$ if $k_{ij}-k_{ji} > 0$,
	or there is no edge between $T_i$ and $T_j$ otherwise.
	The weight of each edge is $\abs{k_{ij} - k_{ji}}$.
	So we have an edge in the direction where the left-to-right order yields fewer intra-column crossings,
	and no edge, if the number of intra-column crossings
	is the same, no matter their~order.
	
	Note that instead of considering each column separately, 
	we can also think of~$G$ as a single graph having 
	at least one connected component for each column.
	
	\paragraph{Implementation and Running Time.}
	For counting the number of crossings between pairs of subtrees,
	we initialize~$\Oh(r^2)$ variables with zero.
	Then, we use a 
	horizontal sweep line traversing all trees in parallel top-down
	while maintaining the current widths
	of $T_1, \dots, T_r$ in variables $b_1, \dots b_r$.
	Our event points are the heights of	the vertices
	(given by $h$) and the heights of the parents
	of the subtree roots (where we have the horizontal
	segment of an edge entering the column).
	
	Whenever we encounter an inter-edge~$uv$,
	we update the crossing variables.
	More precisely, if $u \in T_i$ for some $i \in \set{1, \dots, r}$,
	we consider each $j \in \set{1, \dots, r} \setminus \set{i}$
	and we set $k_{ji} +$$= b_j$ if $v$ is in a column on the left,
	and, symmetrically, we set $k_{ij} +$$= b_j$ if $v$ is in a column on the right.
	
	The running time of this approach is
	$\Oh(r (n_\mathrm{column} + m_\mathrm{column}))$ per column
	where $n_\mathrm{column}$ and $m_\mathrm{column}$
	are the numbers of vertices and edges in the column, respectively.
	Over all columns, this can be accomplished in $\Oh(n^2)$ time.
	
	\paragraph{Size of the Instance.}
	The resulting graph $G$ has $\Oh(n)$
	vertices and $\Oh(n^2)$ edges as we may have an
	edge for every pair of subtrees in a column and we may
	have a linear number of column subtrees.
	The maximum weight $w_{\max}$ in $w$ is in~$\Oh(n^2)$
	since we have $\Oh(n)$ inter-edges and
	the maximum width of a tree is in~$\Oh(n)$.
	
	\paragraph{Comparison of Optima.}
	For each column~$\gamma \in \{1, \dots, \ell\}$,
	we have the following lower bound~$L_\gamma$ 
	on the number~$k_\textrm{column}^\gamma$ of intra-column crossings:
	\begin{equation*}
		L_\gamma = \sum_{i=1}^{r-1} \sum_{j=i+1}^{r} \min\{k_{ij}, k_{ji}\} \le k_\textrm{column}^\gamma.
	\end{equation*}
	
	Now consider a minimum-crossing solution of $\croc{T, h, c}$.
	In column~$\gamma$, the column subtrees have a specific order,
	which we can associate with a permutation~$\Pi^*$ of the column subtrees.
	For simplicity, we rename the column subtrees as $T_1, \dots, T_r$ according to~$\Pi^*$.
	Then, the number of intra-column crossings is
	\begin{equation*}
		k_\textrm{column}^\gamma = \sum_{i=1}^{r-1} \sum_{j=i+1}^{r} k_{ij}.
	\end{equation*}
	Because it is a minimum-crossing solution,
	the number~$\delta_\gamma$ of additional crossings
	(i.e., the deviation of $k_\textrm{column}^\gamma$ from $L_\gamma$)
	due to ``unfavorably'' ordered pairs of subtrees is minimized.
	We can express $\delta_\gamma$ in terms of all $k_{ij}$ by
	\begin{equation*}
		\delta_\gamma = k_\textrm{column}^\gamma - L_\gamma
		= \sum_{i \ne j, k_{ij} > k_{ji}} k_{ij} - k_{ji}.
	\end{equation*}
	
	Now consider a minimum-weight solution $S$ of $\croc{G, w}$
	with weight~$s$.
	After removing~$S$ from~$G$, we have an acyclic graph
	for which we can find a topological order $\Pi$ of its vertices,
	which in turn correspond to the column subtrees in $\croc{T, h, c}$.
	Again, we rename these subtrees as $T_1, \dots, T_r$ according to~$\Pi$.
	If we arrange $G$ linearly according to $\Pi$, only the edges
	of $S$ point backward and,
	by definition of the edge directions,
	for exactly these edges $k_{ij} > k_{ji}$ holds.
	The sum
	\begin{equation*}
		s_\gamma = \sum_{i \ne j, k_{ij} > k_{ji}} k_{ij} - k_{ji}
	\end{equation*}
	of weights in $S$ whose vertices represent pairs of column subtrees in~$\gamma$
	is minimized because~$G$ has independent components for all columns in $c$.
	Therefore, $s_\gamma = \delta_\gamma$.
	Over all columns $\{1, \dots, \ell\}$, we set $t = \sum_{i = 1}^\ell L_i$,
	we have $s = \sum_{i = 1}^\ell s_i$, and we conclude
	\begin{equation*}
		k_\textsf{column} = \sum_{i = 1}^\ell k_\textrm{column}^i
		= \sum_{i = 1}^\ell (L_i + \delta_i) = \sum_{i = 1}^\ell (L_i + s_i) = s + t.
	\end{equation*}
	Note that $t \in \Oh(n^2)$ as we have $\Oh(n^2)$ pairs of edges,
	which cross at most once.
	
	\paragraph{Transforming a Solution Back.}
	Similarly, we can find for any solution $S$ of $\croc{G, w}$ with size~$s$
	a topological order, which corresponds to a subtree arrangement in $\croc{T, h, c}$.
	There, together with the $t$ unavoidable crossings,
	we have at most $s$	additional crossings
	due to ``unfavorably'' ordered pairs.
	We can find the topological order in linear time
	in the size of~$G$, which is in $\Oh(n^2)$.
\end{proof}

When we combine \cref{clm:tc2ifas,clm:ifas2fas}, we obtain \cref{clm:tc2fas}.

\begin{corollary}
	\label{clm:tc2fas}
	We can reduce an instance $\croc{T, h, c}$ of \tcv{2}  on $n$ vertices
	to an instance~$G$ of \fas in time $\Oh(n^4)$ such that 
	\begin{itemize}
		\item $G$ has size in $\Oh(n^4)$,
		\item the number of intra-column crossings in a minimum-crossing solution of $\croc{T, h, c}$ equals
		the size of a minimum-size solution of $G$ plus~$t$
		where $t$ is some integer in $\Oh(n^2)$
		depending only on $\croc{T, h, c}$, and
		\item we can transform a solution of $G$ with size $s$ in $\Oh(n^2)$ time
		to a solution of~$\croc{T, h, c}$ with $s + t$ crossings.
	\end{itemize}
\end{corollary}

\paragraph{Fixed-Parameter Tractable Algorithm.}
For \tcv{2}, one of the most natural parameters is the number of crossings,
which is also the objective value.
With our reduction to 
\fas at hand,
it is easy to show that \tcv{2} is fixed-parameter tractable (FPT) in this parameter for bounded-degree column trees.
This follows from the fact that \fas is FPT in its natural parameter
(the solution size)
as first shown by Chen, Liu, Lu, O'Sullivan, and Razgon~\cite{CLLOR08}.

\begin{theorem} \label{clm:fptv2}
	\tcv{2} is fixed-parameter tractable in the number~$k$ of crossings
	plus~$\Delta$.
	More precisely, given an instance $\croc{T, h, c}$
	with $n$ vertices and maximum vertex degree~$\Delta$,
	there is an algorithm computing an embedding
	with the minimum number $k$ of crossings in $\Oh(\Delta!\Delta n^2 + n^{16} 4^k k^3 k!)$ time.
\end{theorem}
\begin{proof}
	First, we find an embedding of every column subtree
	inducing the minimum number of intra-subtree crossings
	in $\Oh(\Delta!\Delta n^2)$ time 
	(see \cref{clm:algo}).
	
	Employing \cref{clm:tc2fas}, we reduce $\croc{T, h, c}$ in time
	$\Oh(n^4)$ to an instance~$G$ of \fas of size $\Oh(n^4)$.
	We compute a minimum-size solution $S$ for the \fas instance~$G$
	in time $\Oh((n^4)^4 4^s s^3 s!)$ where $s = |S|$~\cite{CLLOR08}.
	In $\Oh(n^2)$ time, we use $S$ to compute
	a subtree arrangement of $\croc{T, h, c}$.
	We return the resulting embedding.
	
	This solution has $k = k_\mathsf{subtree} + k_\mathsf{column} + k_\mathsf{inter}$
	crossings,
	where $k_\mathsf{subtree}$ and $k_\mathsf{column}$
	are minimum due to \cref{clm:algo} and \cref{clm:tc2fas}, respectively.
	As $k_\mathsf{inter}$ is always the same,
	the resulting embedding has the minimum number of crossings.
	
	Overall, this algorithm runs in $\Oh(\Delta!\Delta n^2 + n^4 + (n^4)^4 4^s s^3 s! + n^2)$ time.
	Since $s \le k$, this is an FPT algorithm in $k + \Delta$.
\end{proof}

\paragraph{Approximation Algorithm.}
Similar to our FPT result, we can use any approximation algorithm
for \fas to approximate \tcv{2}. 
It is an unresolved problem 
whether \fas admits a constant-factor approximation. 
In case such an approximation is found,
this immediately propagates to \tcv{2}.
Currently, we can employ the best known approximation
algorithm of \fas due to Even, Naor, Schieber, and Sudan~\cite{ENSS98},
which has an approximation factor of $\Oh(\log n' \log \log n')$,
where $n'$ is the number of vertices in the \fas instance.
We remark that this algorithm involves solving a linear program
and the authors do not write much about precise running time bounds,
which we also avoid here.

\begin{theorem} \label{clm:approx}
	There is an approximation algorithm that,
	for a given instance $\croc{T, h, c}$ of \tcv{2}
	with $n$ vertices and maximum degree~$\Delta$,
	computes in $\Oh(\mathsf{poly}(n) \cdot \Delta!\Delta)$ time
	an embedding where the number of crossings
	is at most $\Oh(\log n \log \log n)$ times the minimum number of crossings.
\end{theorem}
\begin{proof}
	First, we find an embedding of every column subtree
	inducing the minimum number of intra-subtree crossings
	in $\Oh(\Delta!\Delta n^2)$ time 
	(see \cref{clm:algo}).

	Employing \cref{clm:tc2fas}, we reduce $\croc{T, h, c}$
	in polynomial time to an instance~$G$ of \fas of size $\Oh(n^4)$.
	Using the algorithm by Even et al.~\cite{ENSS98},
	we compute an approximate solution $S$ for the \fas instance~$G$
	where $|S|$ is greater than an optimal solution by at
	most a factor in $\Oh(\log n^4 \log \log n^4) = \Oh(\log n \log \log n)$.
	From $S$, we compute a subtree arrangement in polynomial time.
	We return the resulting embedding.
	
	In this embedding, the number of inter-column crossings
	and the number of intra-subtree crossings is minimum,
	and only the number of intra-column crossings is approximated.
	Therefore, the approximation factor of
	$\Oh(\log n \log \log n)$ holds for the total number~$k$
	of crossings as well.
\end{proof}

\section{Heuristic Approaches for Harder Versions} 
\label{sec:hardest}
Here, we briefly discuss how to approach harder variants of the \tc problem,
specifically under V3 or with a variable column order.

\paragraph{\tcv{3}.}
Recall that, compared to V2, in V3 column subtrees may interleave
within a column.
By \cref{clm:np}, the \tcv{3} problem is NP-hard even for binary column trees
and the proof does not require interleaving.
However, as shown in~\cref{fig:variants}, 
interleaving can help to reduce the number of crossings.
So while the two tasks of subtree embedding and subtree arrangement cannot be considered separately,
it is still natural to try this as a heuristic.

A greedy approach for \tcv{3} could work as follows. 
First, find an optimal embedding for each column subtree
using the algorithm from \cref{clm:algo}.
Second, for each column,
sort the column trees in descending order based on their root heights.
Then, in this order, add one column tree at a time 
at the horizontal position where it induces
the smallest number of new crossings.
We leave the details on how to find this position open for now.
And while we think that this strategy can work well in practice,
we want to point out that, in general, it can
cause linearly more crossings than the optimum:
The two column trees in~\cref{fig:greedyBad} can be drawn 
without intra-subtree crossings (\cref{fig:greedyBad:one,fig:greedyBad:two,fig:greedyBad:three}),
yet even with interleaving we get a linear number of (intra-column) crossings.
However, at the cost of two intra-subtree crossings,
we can reduce the total number of crossings to four; see~\cref{fig:greedyBad:four}.

\begin{figure}[tbh]
	\centering
	\begin{subfigure}[t]{0.23 \linewidth}
		\centering
		\includegraphics[page=1]{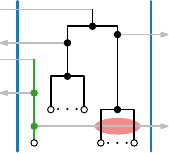}
		\caption{}
		\label{fig:greedyBad:one}
	\end{subfigure}
	\hfill
	\begin{subfigure}[t]{0.23 \linewidth}
		\centering
		\includegraphics[page=2]{greedyBad}
		\caption{}
		\label{fig:greedyBad:two}
	\end{subfigure}
	\hfill
	\begin{subfigure}[t]{0.23 \linewidth}
		\centering
		\includegraphics[page=3]{greedyBad}
		\caption{}
		\label{fig:greedyBad:three}
	\end{subfigure}
	\hfill
	\begin{subfigure}[t]{0.23 \linewidth}
		\centering
		\includegraphics[page=4]{greedyBad}
		\caption{}
		\label{fig:greedyBad:four}
	\end{subfigure}
	\caption{With $x \approx n/2$, computing optimal subtree embeddings first (a--c) can cause linearly more crossings
		than using a non-optimal subtree embedding and interleaving (d). (Each of the schematic subtrees has $x$ leaves.)} 
	\label{fig:greedyBad}
\end{figure}

\paragraph{Variable Column Order.}
In many cases in practice, the column order can be assumed to be fixed
as there is, for example, a natural order for age groups
and even for geographic regions of a country, 
there might be a preferred or standard ordering.
Let us nonetheless assume now that the column order is not fixed but variable.
We want to remark that we can still get the same results from \cref{sec:algos} (under V2)
if we add another parametric dependence on the number~$\ell$ of columns;
we can try all $\ell!$ permutations of the columns and apply these algorithms to each of them.
While we can expect $\ell$ in practice to be low, 
this becomes in general, of course, quickly infeasible.
Alternatively, we might try to find a column order that minimizes the number of inter-column crossings,
for example with a greedy algorithm.
The three different tasks, namely, finding a column order, subtree embeddings, and subtree arrangement,
would thus focus on the three different crossings~types.

\section*{Acknowledgments}
We thank the anonymous reviewers for their helpful comments.

\pdfbookmark[1]{References}{References}
\bibliographystyle{plainurl}
\bibliography{sources}

\end{document}

%% file: config.tex
\usepackage[T1]{fontenc}
\usepackage{amsmath,amssymb}
\usepackage{xspace}
\usepackage[usenames,dvipsnames,table]{xcolor}
\usepackage{wrapfig,graphicx}
\DeclareGraphicsExtensions{.pdf,.png}
\graphicspath{{figs/}{notes/}}
\usepackage{subcaption}
\usepackage[inline]{enumitem}

\definecolor{dark blue}{rgb}{0.121,0.47,0.705}

\definecolor{dark red}{rgb}{0.89,0.102,0.109}

\definecolor{dark orange}{rgb}{1,0.498,0}

\definecolor{dark green}{rgb}{0.2,0.627,0.172}

\definecolor{dark purple}{rgb}{0.415,0.239,0.603}

\definecolor{dark brown}{rgb}{0.651, 0.337, 0.157}

\let\emph\relax
\DeclareTextFontCommand{\emph}{\color{dark blue}\em}

\usepackage{mathtools}
\DeclarePairedDelimiter\set{\lbrace}{\rbrace}
\DeclarePairedDelimiter\abs{\lvert}{\rvert}

\DeclarePairedDelimiter\croc{\langle}{\rangle}

\def\Oh{\ensuremath{\mathcal{O}}}

\newcommand{\tc}{\textsc{TreeColumns}\xspace}
\newcommand{\tcv}[1]{\textsc{TreeColumns|V#1}\xspace}
\newcommand{\fas}{\textsc{FAS}\xspace}
\newcommand{\ifas}{\textsc{IFAS}\xspace}

\usepackage{thmtools} 
\usepackage{thm-restate}
\usepackage{apptools}
\newcommand{\restateref}[1]{\IfAppendix{\hyperref[#1]{$\star$}}{\hyperref[#1*]{$\star$}}}
\usepackage{url}
\usepackage{cite}
\usepackage[hidelinks=true,
			bookmarks=true,
			bookmarksopen=true,
			bookmarksopenlevel=2,
			bookmarksnumbered=true,
			hyperindex=true,
			plainpages=false,
			pdfpagelabels=true
]{hyperref} 
\usepackage[capitalise,nameinlink]{cleveref}
\crefname{observation}{Observation}{Observations}
\Crefname{observation}{Observation}{Observations}
\crefname{proposition}{Proposition}{Propositions}
\Crefname{proposition}{Proposition}{Propositions}
\newcommand{\etal}{{et~al.}\xspace}

\renewcommand{\orcidID}[1]{\href{https://orcid.org/#1}{\includegraphics[scale=.03]{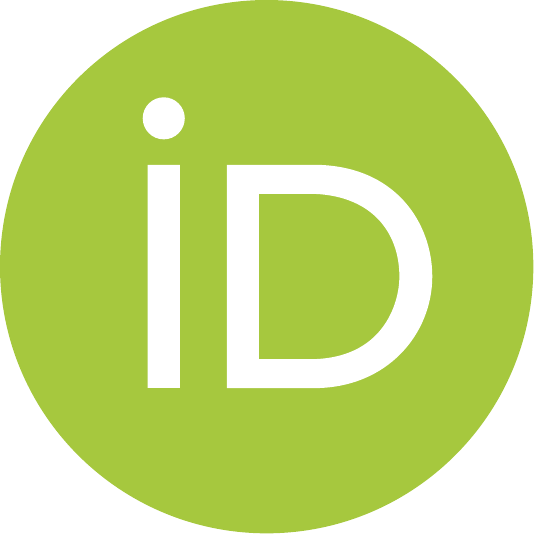}}}

\let\doendproof\endproof
\renewcommand\endproof{~\hfill$\qed$\doendproof}